\documentclass[a4paper,twocolumn,10pt,accepted=2026-06-07]{quantumarticle}
\pdfoutput=1
\usepackage[utf8]{inputenc}
\usepackage[english]{babel}
\usepackage[T1]{fontenc}
\usepackage[numbers,sort&compress]{natbib}
\usepackage{amsmath,amsthm,amsfonts,amssymb,amsbsy,mathrsfs,mathdots,graphicx,xcolor,xfrac,mathtools,enumerate,xr,subfigure,bm,bbm,verbatim,comment,dsfont,soul,array,multirow,graphicx,multibib,enumitem}

\usepackage[colorlinks]{hyperref}
\makeatletter
\newcommand\org@hypertarget{}
\let\org@hypertarget\hypertarget
\renewcommand\hypertarget[2]{%
  \Hy@raisedlink{\org@hypertarget{#1}{}}#2%
  }
\makeatother
\hypersetup{
	bookmarksnumbered,
	pdfstartview={FitH},
	citecolor={darkgreen},
	linkcolor={darkred},
	urlcolor={darkblue},
	pdfpagemode={UseOutlines}}
\definecolor{darkgreen}{RGB}{50,190,50}
\definecolor{darkblue}{RGB}{0,0,190}
\definecolor{darkred}{RGB}{238,0,0}
\definecolor{darkpurple}{rgb}{0.4, 0.2, 0.5}

\newcommand{\ket}[1]{\ensuremath{\left|\right.\!{#1}\!\left.\right\rangle}}

\newcommand{\ketbra}[2]{\ensuremath{|{#1}\rangle\!\langle{#2}|}}

\newcommand{\id}{\mathds{1}}
\DeclareMathOperator{\Tr}{Tr}

\newtheorem{coro}{Corollary}
\newtheorem{prop}{Proposition}

\begin{document}
\title{Exploring Imaginary Coordinates: Disparity in the Shape of Quantum State Space in Even and Odd Dimensions}
\author{Simon Morelli}
\email{simon.morelli@tuwien.ac.at}
\affiliation{Technische Universität Wien, Atominstitut, Stadionallee 2, 1020 Vienna, Austria}
\affiliation{BCAM - Basque Center for Applied Mathematics,
Mazarredo 14, E48009 Bilbao, Basque Country - Spain}
\orcid{0000-0002-7588-2701}
\author{Santiago Llorens}
\email{santiago.llorens@ug.edu.pl}
\affiliation{Division of Quantum Information, Faculty of Mathematics, Physics and Informatics,
University of Gda{\'n}sk, Wita Stwosza 57, 80-308 Gda{\'n}sk, Poland}
\affiliation{F{\'i}sica Te{\`o}rica: Informaci{\'o} i Fen{\`o}mens Qu{\`a}ntics, Departament de F{\'i}sica, Universitat Aut{\`o}noma de Barcelona, 08193 Bellaterra, Spain}
\orcid{0000-0002-8643-447X}

\author{Jens Siewert}
\email{jens.siewert@ehu.eus}
\affiliation{Department of Physical Chemistry and EHU Quantum Center, University of the Basque Country UPV/EHU,
E-48080 Bilbao, Spain}
\affiliation{IKERBASQUE Basque Foundation for Science, E-48009 Bilbao, Spain}
\orcid{0000-0002-9410-5043}


\begin{abstract}
The state of a finite-dimensional quantum system is described by a density matrix that can be decomposed into a real diagonal, a real off-diagonal and and an imaginary off-diagonal part. The latter plays a peculiar role. While it is intuitively clear that some of the imaginary coordinates cannot have the same extension as their real counterparts the precise relation is not obvious. We give a complete characterization of the constraints in terms of tight inequalities for real and imaginary Bloch-type coordinates. Our description entails a three-dimensional Bloch ball-type model for the state space. We uncover a surprising qualitative difference for the state-space boundaries in even and odd dimensions.
\end{abstract}

\maketitle

\section{Introduction}

From the very beginning, quantum mechanics was naturally represented in terms of complex numbers. The state of a quantum system was identified with an element of a complex Hilbert space~\cite{vNeumann_1927,Dirac_1930}.
Whether this complex Hilbert-space structure is actually necessary or merely more elegant than a hypothetical real version was an open question until very recently. 
Assuming the dimension of the system to be bounded, it is not difficult to design experiments that show that states in a complex Hilbert space are required to describe the observed data. If the dimension is bounded by two, a sequential Stern-Gerlach experiment will suffice to come to this conclusion~\cite{Sakurai_Napolitano_2020}.
However, if no constraints on the dimension are assumed, it is always possible to explain the measurement result of a single system by a real system of higher dimension. This is even possible for two parties that are restricted to local operations~\cite{Pal_2008,McKague_2009}.
Only recently it was shown that real quantum theory, that is, quantum theory with states belonging to real Hilbert spaces with the standard tensor structure, can be experimentally refuted by designing an experiment in a quantum network involving three parties~\cite{Renou_2021}.

Thus, while it appears that a complex Hilbert space structure is intrinsically needed to describe quantum mechanics, the role of 'imaginarity' has yet to be fully understood, cf., e.g., Refs.~\cite{Caves2001,Wootters2012,Chen_2023,Xu_2023,Fernandes2024,Bruss2026,Woods2026}.
One might ask 'how real' or 'how imaginary' a quantum state can be. The resource theory of imaginarity~\cite{Hickey_2018, Wu_2021,Wu_2021_2} measures imaginarity as a quantity that cannot be increased under the set of free operations, real quantum channels.
We approach this fundamental question from a different perspective. In analogy with the qubit Bloch ball description based on the three Pauli matrices~\cite{Bengtsson_2006} we define a diagonal, a real off-diagonal and a purely imaginary coordinate and ask what weight can be put on these coordinates to be compatible with a positive operator. We then introduce two novel inequalities in Proposition~\ref{prop:i<1+r} and~\ref{prop:linear_bound}, which completely characterize the set of attainable values for these coordinates and obtain a Bloch ball-type model~\cite{Bengtsson_2012} of the quantum state space.

The Bloch representation of a qubit suggests the quantum space state to be rotation invariant with respect to these three  coordinates. As it turns out, this only holds for the real diagonal and real off-diagonal coordinate in arbitrary dimensions, while the imaginary coordinate has to be treated separately.
In this article we give a precise quantitative bound on the imaginary Bloch length of a quantum state compared to both its real diagonal and real off-diagonal Bloch components. Surprisingly, we thereby find a qualitative difference between systems of even dimensions and systems of odd dimensions. This difference is most pronounced for small system size and vanishes in the asymptotic limit as one might expect. It is noteworthy that the generic form for systems of odd dimension appears only from dimension 5 onward.

\section{Notation}

Every quantum state can be written in the following way,
\begin{align}
    \rho\ =\ \frac{1}{d}(\id_d+D+X+I)\ =\ \frac{1}{d}(R+I)\ \ ,
\label{eq:DXImat}
\end{align}
where $R$ is a real, positive and symmetric matrix, $D$ is a diagonal matrix, $X$ a real off-diagonal matrix and $I$ a purely imaginary (off-diagonal) and skew-symmetric matrix. The operators $D$, $X$ and $I$ are all Hermitian, traceless, and orthogonal to each other, that is $\Tr(D X)=\Tr(D I)=\Tr(X I)=0$.
For dimension $d=2$ we recover the Bloch representation in terms of the Pauli-matrices, where $D=\Tr\left(\rho\sigma_Z\right)\sigma_Z\equiv\langle\sigma_Z\rangle\sigma_Z$, $X=\langle\sigma_X\rangle\sigma_X$ and $I=\langle\sigma_Y\rangle\sigma_Y$.
Let us introduce the quantities
\begin{align*}
    S_{\text{D}}=\sqrt{\frac{\Tr D^2}{d}},\quad S_{\text{X}}=\sqrt{\frac{\Tr X^2}{d}},\quad S_{\text{I}}=\sqrt{\frac{\Tr I^2}{d}}
\end{align*}
as magnitudes of the diagonal, real off-diagonal and the imaginary part of the state $\rho$. For a qubit we recover $S_{\text{D}}=\langle\sigma_Z\rangle$, $S_{\text{X}}=\langle\sigma_X\rangle$ and $S_{\text{I}}=\langle\sigma_Y\rangle$.
It holds that
 \begin{align}\label{eq:purity}
     \Tr\rho^2&\ =\ \frac{1}{d^2}(d+\Tr D^2+\Tr X^2+\Tr I^2)\notag\\
     &\ =\ \frac{1}{d}\left(1+S_{\text{D}}^2+S_{\text{X}}^2+S_{\text{I}}^2\right)\ \ .
 \end{align}
In dimensions $d>2$ the quantities $S_{\text{D}}$, $S_{\text{X}}$ and $S_{\text{I}}$ can be viewed as a coarse-graining of the Bloch vector of $\rho$.
Every quantum state $\rho$ can be expanded in an orthogonal matrix basis as
\begin{align}
     \rho\ =\ \frac{1}{d}\left(\id_d+\sum\limits_{k=1}^{d^2-1} v_k \mu_k\right)\ ,
\label{eq:blochd}
\end{align}
where we call a basis $\{\mu_k\}_{i=0}^{d^2-1}$ satisfying $\mu_0=\id_d$ and 
$\Tr(\mu_k\mu_l^{\dagger})=d\delta_{kl}$ a Bloch basis~\cite{BertlmannKrammer08,Siewert_2022}.
The vector $[\mathbf{v}]_k=v_k=\langle\mu_k\rangle$ is called the Bloch vector of the state $\rho$, where $v_0=1$ is not included.
Choosing a matrix basis that includes only diagonal elements $\mu^\text{d}_k$, strictly real off-diagonal elements $\mu^\text{x}_k$ and strictly imaginary elements $\mu^\text{i}_k$, e.g., the generalized Gell-Mann matrices, the state can be expanded as
\begin{align}
     \rho = \frac{1}{d}\scalebox{0.95}{$\left(\id_d+\sum\limits_{k=1}^{d-1} v^\text{d}_k \mu^\text{d}_k+\sum\limits_{k=1}^{\scriptstyle{\tfrac{d}{2}(d-1)}} v^\text{x}_k \mu^\text{x}_k+\sum\limits_{k=1}^{\scriptstyle{\tfrac{d}{2}(d-1)}} v^\text{i}_k \mu^\text{i}_k\right)$},
\label{eq:bloch_real_imaginary}
\end{align}
where
\begin{align*}
    \scalebox{0.95}{$
    D=\sum\limits_{k=1}^{d-1} v^\text{d}_k \mu^\text{d}_k\ ,\quad
    X=\sum\limits_{k=1}^{\scriptstyle{\tfrac{d}{2}(d-1)}} v^\text{x}_k \mu^\text{x}_k\ ,\quad
    I=\sum\limits_{k=1}^{\scriptstyle{\tfrac{d}{2}(d-1)}} v^\text{i}_k \mu^\text{i}_k\ .
    $}
\end{align*}
It follows immediately that 
\begin{align*}
    \scalebox{0.95}{$
    S_{\text{D}}^2=\sum\limits_{k=1}^{d-1} |v^\text{d}_k|^2\ ,\quad
    S_{\text{X}}^2=\sum\limits_{k=1}^{\scriptstyle{\tfrac{d}{2}(d-1)}} |v^\text{x}_k|^2\ ,\quad
    S_{\text{I}}^2=\sum\limits_{k=1}^{\scriptstyle{\tfrac{d}{2}(d-1)}} |v^\text{i}_k|^2\ .
    $}
\end{align*}

\section{Bounds}

We now investigate which values of $S_{\text{D}}$, $S_{\text{X}}$ and $S_{\text{I}}$ are compatible with a quantum state.
The purity $\Tr\rho^2$ of a quantum state is bounded by one and therefore from Eq.~\eqref{eq:purity} it follows that,
\begin{align}\label{eq:purity_bound} 
    S_{\text{D}}^2+S_{\text{X}}^2+S_{\text{I}}^2\ \le\ d-1\ \ .
\end{align}
For qubits this is the only restriction and the quantum set, that is, the set $\{(S_{\text{D}}(\rho), S_{\text{X}}(\rho),S_{\text{I}}(\rho):\rho\ge0,\ \Tr\rho=1\}$, is rotation invariant. We will see that the same applies in higher dimension only for the values $S_{\text{D}}$ and $S_{\text{X}}$,  while the coordinate $S_{\text{I}}$ plays a special role.
We therefore start our investigation by bunching the two real coordinates together to $S_{\text{R}}=\sqrt{S_{\text{D}}^2+S_{\text{X}}^2}$. It turns out that the constraints we get for the $S_{\text{R}}$ and $S_{\text{I}}$ values are the only ones that exist, and that the real coordinates
are subject to rotational symmetry.

Every quantum state $\rho$ can be split into its real and imaginary parts as $\rho=(R+I)/d$,
with $R=d(\rho+\rho^T)/2$ and $I=d(\rho-\rho^T)/2$~\cite{Hickey_2018}.
The operator $R=\id+D+X$ with $\Tr R=d$ is orthogonal to the imaginary part, $I$.
Note that, while $R$ includes the identity $\id$, the length $S_{\text{R}}$ does not contain its contribution. That is,  
\mbox{$S_{\text{R}}^2=\Tr R^2/d-1$}.

\begin{figure*}
\begin{center}
     \includegraphics[width=0.7\columnwidth]{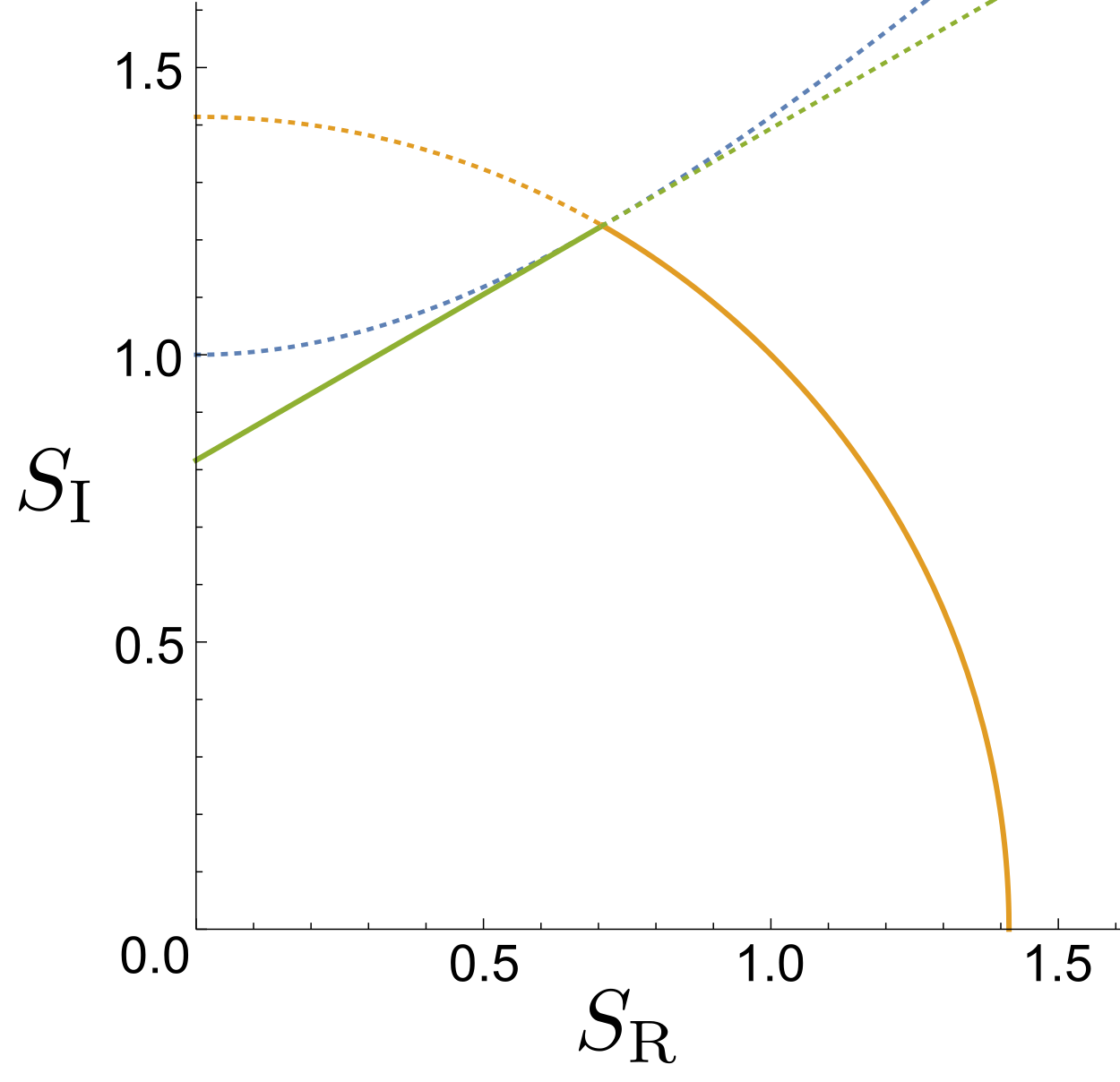}
     \hspace{-0.05\columnwidth}
     \includegraphics[width=0.7\columnwidth]{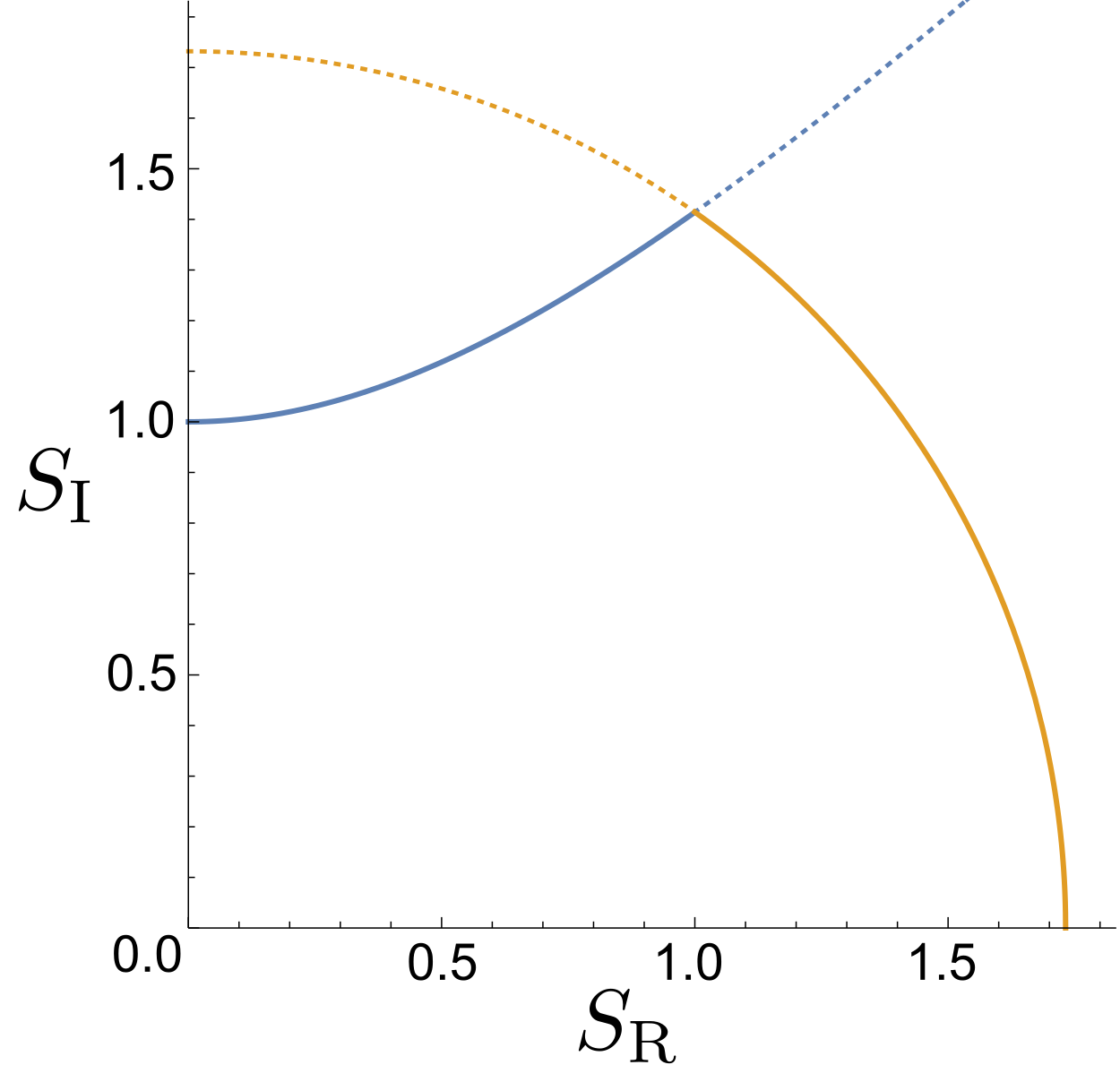}
     \hspace{-0.05\columnwidth}
     \includegraphics[width=0.7\columnwidth]{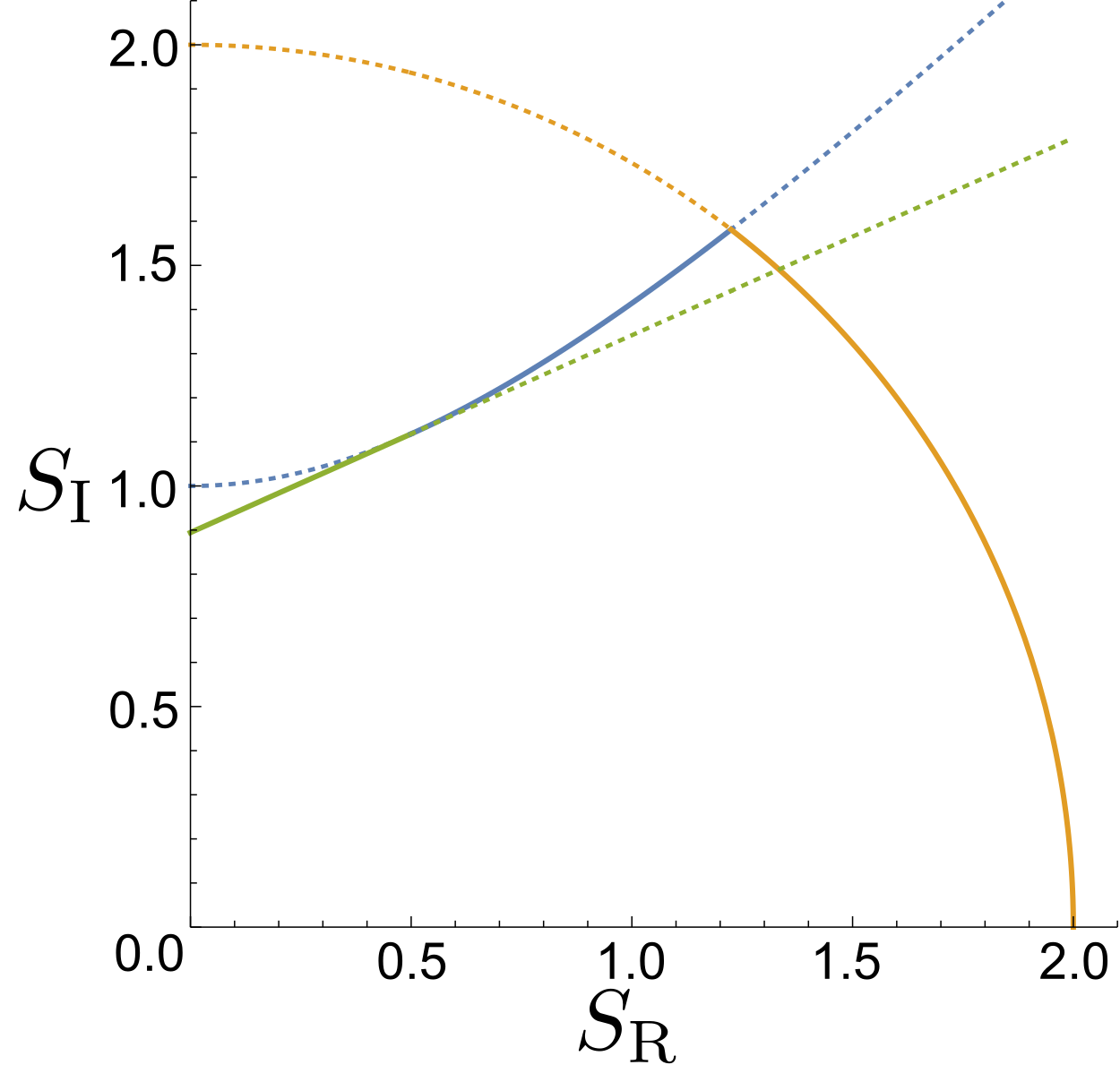}\\
     (a) \hspace{0.6\columnwidth} (b) \hspace{0.6\columnwidth} (c)\\
        \caption{The three figures display the purity bound of Eq.~\eqref{eq:purity_bound} in orange, the quadratic bound of Eq.~\eqref{eq:quadratic_bound} in blue and the linear bound of Eq.~\eqref{eq:linear_bound} in green, (a) $d=3$, (b) $d=4$, (c) $d=5$. Continuous lines show the actual bounds, dotted lines are their continuation, where they are either not valid or not tight.
        }
    \label{fig:areas}
\end{center}
\end{figure*}

\begin{prop}
For every finite-dimensional quantum state it holds that the quantities $S_{\mathrm{R}}$ and $S_{\mathrm{I}}$ satisfy
\begin{align}\label{eq:quadratic_bound}
    S_{\mathrm{I}}^2\ \le\ 1 + S_{\mathrm{R}}^2\ \ .
\end{align}
\label{prop:i<1+r}
\end{prop}

\begin{proof}
It follows from the positivity of $\rho$ and $\rho^T$ that
\begin{align}
    0&\le d \Tr\left(\rho\rho^T\right)=\frac{1}{d}\Tr[(R+I)(R-I)]\\
    &=\frac{1}{d}(\Tr R^2-\Tr I^2)=1+S_{\text{R}}^2-S_{\text{I}}^2\ \ .\notag
\end{align}
\end{proof}

The bound is saturated by states of the form
\begin{align}\label{states_even}
    \rho\ =  
    \scalebox{0.75}{$
    \begin{pmatrix} 
    \alpha_2 & -i\alpha_2 &&&&&\\
    i\alpha_2 i& \alpha_2 &&&&&\\
    && \alpha_4 & -i\alpha_4  &&&\\
    && i\alpha_4  & \alpha_4 &&&\\
    &&&& \ddots &&&\\
    &&&&& \alpha_{d} & -i\alpha_{d}  \\
    &&&&& i\alpha_{d}  & \alpha_{d}
    \end{pmatrix}
    $}
\end{align}
for even dimension and states of the form
\begin{align}\label{states_odd_1}
    \rho\ =
    \scalebox{0.75}{$
    \begin{pmatrix} 
    \alpha_2 & -i\alpha_2 &&&&&&\\
    i\alpha_2 & \alpha_2 &&&&&&\\
    && \alpha_4 & -i\alpha_4  &&&&\\
    && i\alpha_4  & \alpha_4 &&&&\\
    &&&& \ddots &&&\\
    &&&&& \alpha_{d-1} & -i\alpha_{d-1}  &\\
    &&&&& i\alpha_{d-1}  & \alpha_{d-1} &\\
    &&&&&&& 0
    \end{pmatrix}
    $}
\end{align}
for odd dimensions if $S_{\text{R}}\ge 1/\sqrt{d-1}$. This is easy to see by noting that in both cases $\Tr R^2=\Tr I^2=d^2\sum\alpha_k^2$. This bound is shown in blue in Figs.~\ref{fig:areas}--\ref{fig:asymptotic}.

\begin{prop}
For a quantum state in odd dimension it holds that the quantities $S_{\mathrm{R}}$ and $S_{\mathrm{I}}$ satisfy
\begin{align}\label{eq:linear_bound}
    	\sqrt{d}\ S_{\mathrm{I}}\ \le\ \sqrt{d-1}\ + \ S_{\mathrm{R}}\ \ ,
\end{align}
whenever $S_{\mathrm{R}}\le1/\sqrt{d-1}$.
\label{prop:linear_bound}
\end{prop}

Before proving this bound in full generality, we first show the case $S_{\text{R}}=0$.
Every state satisfying this can be written as
\begin{align}
    \rho\ =\ \frac{1}{d}(\id+I)\ \ .
\end{align}
The operator $\id+I$ needs to be positive, that is, its eigenvalues need to be positive. Let $U$ be a unitary diagonalizing $I$. Since unitary transformations do not change the eigenvalues we can write
\begin{align}
    U(\id+I)U^\dagger\ =\ \id+UIU^\dagger\ =\ \id+I_{\text{diag}}\ \ .
\end{align}
Therefore it must hold for the eigenvalues of $I$ that $|\lambda_k(I)|\le 1$. Since $I$ is traceless it further holds that $\sum\lambda_k(I)=0$. The quantity $\Tr I^2=\sum\lambda_k(I)^2$ is convex, hence maximizing it gives a solution on the boundary.
Under the constraints $\sum\lambda_k(I)=0$ $|\lambda_k(I)|\le 1$ we find the optimal solution to be $\lambda_{2k-1}(I)=-\lambda_{2k}(I)=1$ with $\lambda_{d}(I)=0$ for odd dimension. Therefore,
\begin{align}
    S_{\text{I}}^2\ =\ \frac{1}{d}\Tr I^2\ \le\  \frac{2}{d}\lfloor\frac{d}{2}\rfloor\ \ ,
\end{align}
so that $S_{\text{I}}^2\le 1$ for even dimension and $S^2_{\text{I}}\le (d-1)/d$ otherwise.

\begin{figure*}
\begin{center}
     \includegraphics[width=0.7\columnwidth]{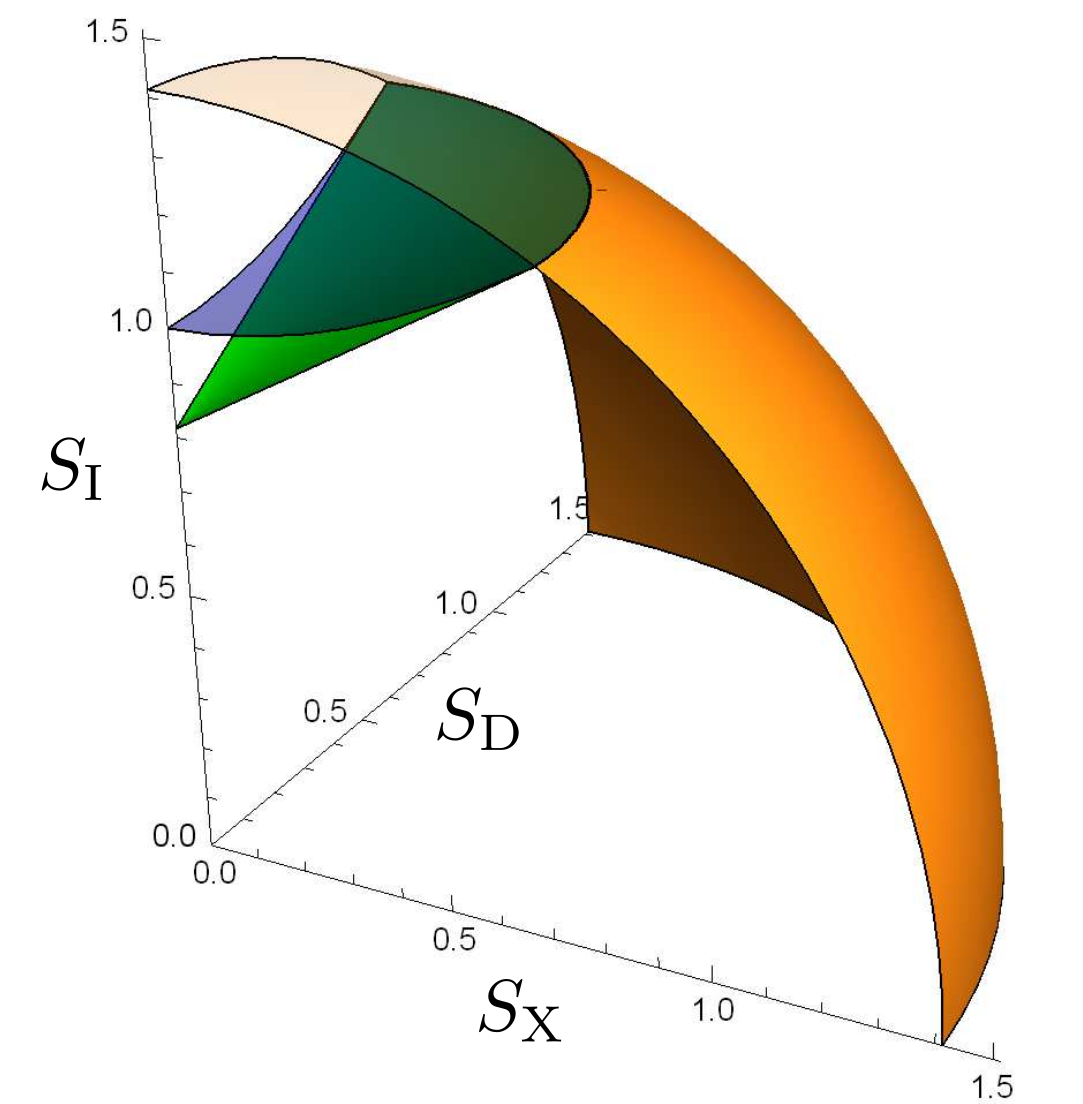}
     \hspace{-0.05\columnwidth}
     \includegraphics[width=0.7\columnwidth]{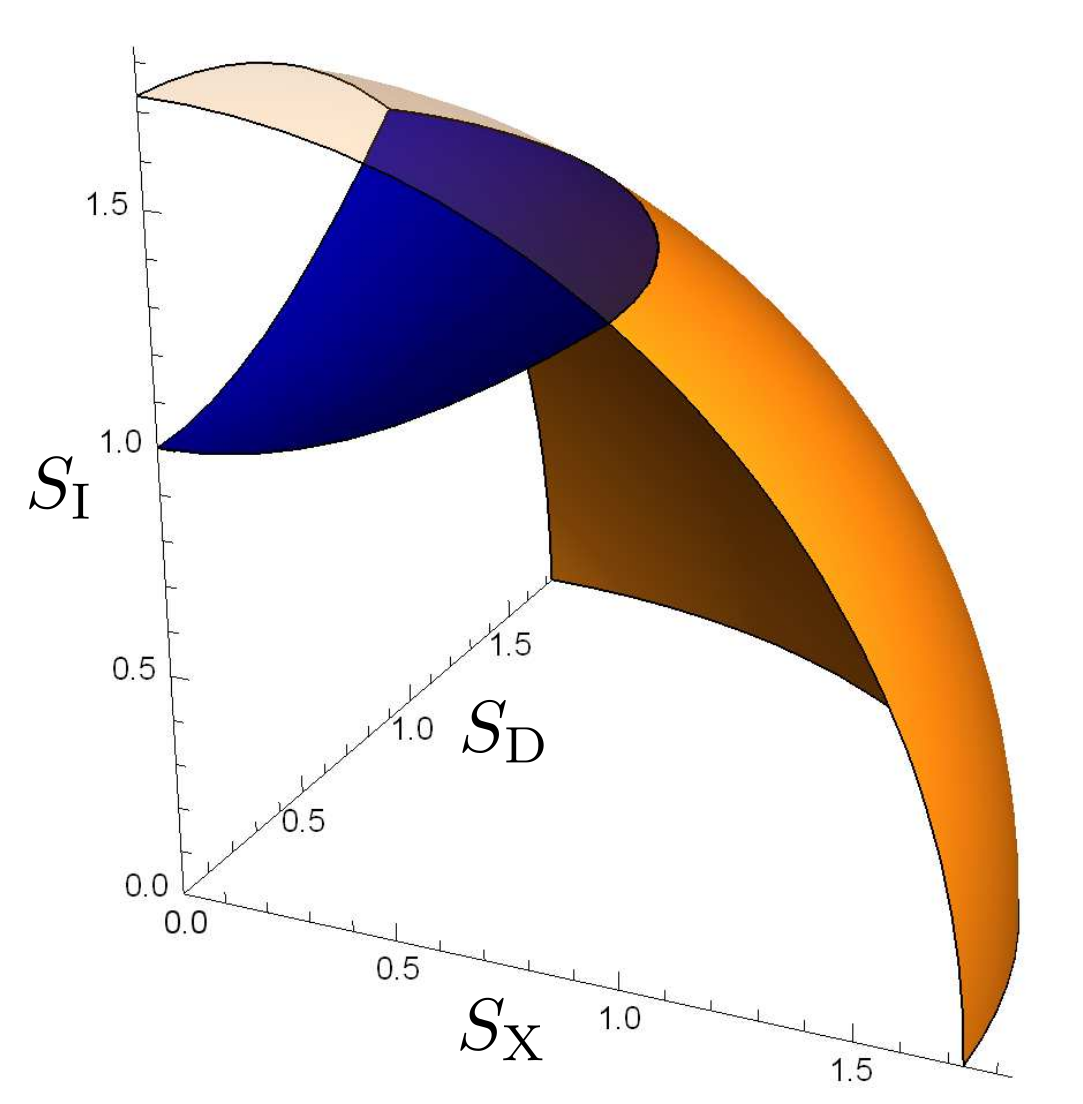}
     \hspace{-0.05\columnwidth}
     \includegraphics[width=0.7\columnwidth]{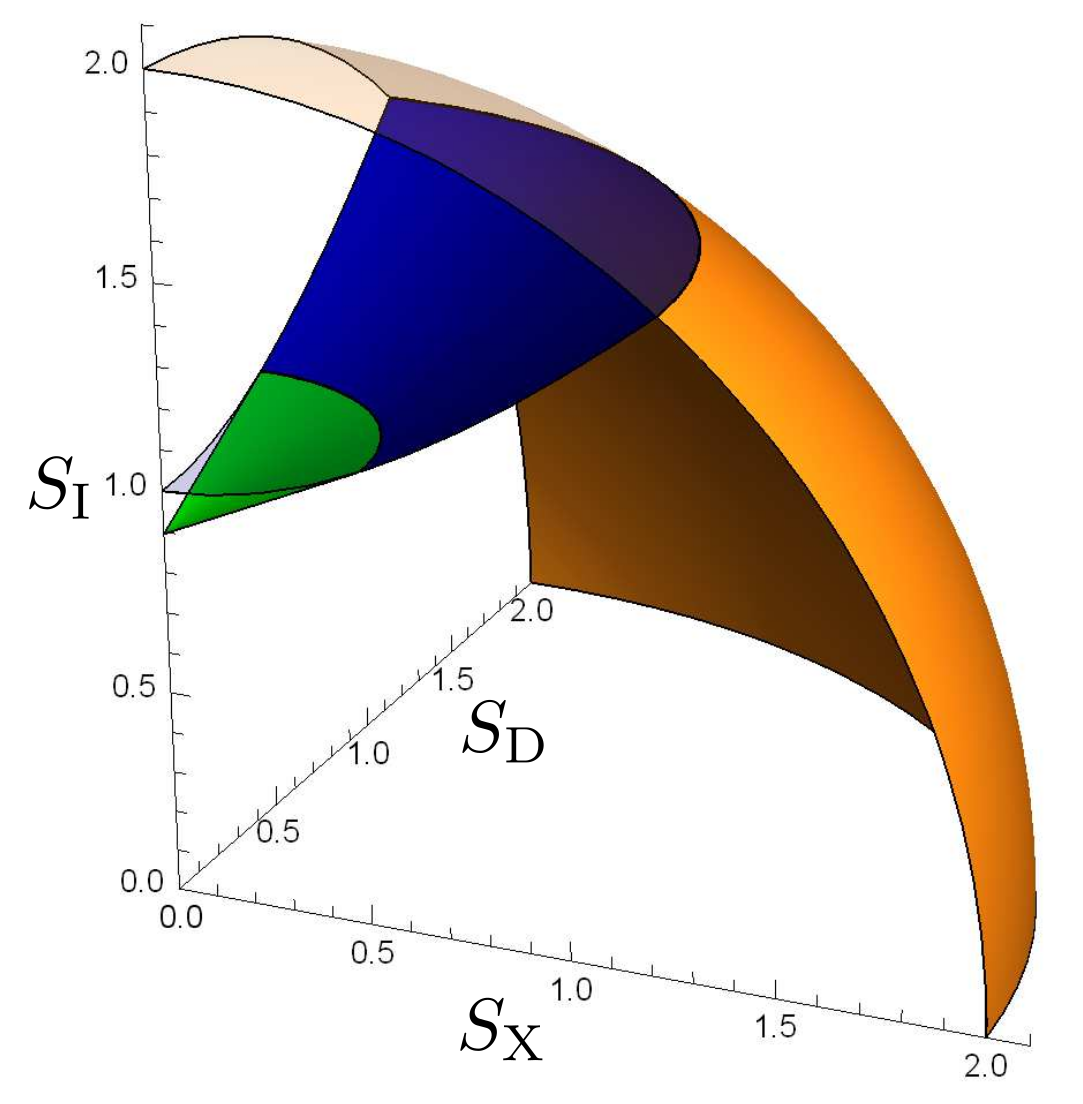}\\
     (a) \hspace{0.6\columnwidth} (b) \hspace{0.6\columnwidth} (c)\\
        \caption{The parametrization of qudit states in Eqs.~\eqref{eq:DXImat}--\eqref{eq:bloch_real_imaginary} defines a three-dimensional model of the state space. The pure states saturate the bound in Eq.~\eqref{eq:purity_bound} (shown in orange). They are located on the surface of the model and have a minimum real coordinate $S_{\text{R}}\geq \sqrt{(d-2)/2}$. This is because there is no pure-state projector corresponding to a superposition with (nearly) purely imaginary off-diagonal elements. The shape of the mixed-state surface part close to the imaginary axis varies as a function of the dimension. (a) $d=3$: The boundary is described by the linear bound Eq.~\eqref{eq:linear_bound} (shown in green), the quadratic bound Eq.~\eqref{eq:quadratic_bound} (shown in blue) is not tight for any mixed state.
        (b) $d=4$: For even dimensions the surface is given by the quadratic bound Eq.~\eqref{eq:quadratic_bound} (blue).
        (c) $d=5$: The generic case for odd dimensions where for small $S_{\text{R}}$ the surface is conic Eq.~\eqref{eq:linear_bound} (green) and for 
        $S_{\text{R}}\geq 1/\sqrt{d-1}$ saturates the quadratic bound Eq.~\eqref{eq:quadratic_bound} (blue).
        }
    \label{fig:areas3d}
\end{center}
\end{figure*}

Subsequently we prove Proposition~\ref{prop:linear_bound} for the general case.
\begin{proof}
We know that $I$ is a skew-symmetric matrix $A^{\intercal}=-A$. Since every square matrix has the same eigenvalues as its transpose, the eigenvalues of $I$ come in pairs with opposite sign, with an additional 0 for odd dimensions. Let $\lambda_k(I)$ be ordered such that $\lambda_{2k}(I)=-\lambda_{2k-1}(I)$.
Let $U$ be a unitary diagonalizing $I=UI_{\text{diag}}U^\dagger$ and define $\Tilde{R}=U^\dagger RU$.
As states remain positive under transposition, we know that both $R+I$ and $R-I$ are positive matrices.
It then also holds that $\Tilde{R}\pm I_{\text{diag}}\ge0$ and it follows that $\Tilde{R}_{kk}\ge|\lambda_{k}(I)|$.
It is known that the diagonal of a matrix $\Tilde{R}_{kk}$ is majorized by its eigenvalues $\lambda_k(\Tilde{R})=\lambda_k(R)$. Since the squared sum is a Schur convex function, it holds that $\sum\limits_{k=1}^{d}\Tilde{R}_{kk}^2\le\sum\limits_{k=1}^{d}\lambda_{k}(R)^2$
and we conclude
\begin{align*}
    \Tr I^2&=\sum\lambda_k(I)^2\le\sum\limits_{k=1}^{d-1}\Tilde{R}_{kk}^2
    \le\sum\limits_{k=1}^{d-1}\lambda_{k}(R)^2\le\Tr R^2\ .
\end{align*}
With this we recover the result of Proposition~\ref{prop:i<1+r}. The first and the second inequality are saturated if all eigenvalues satisfy $\lambda_{2k-1}(I)=-\lambda_{2k}(I)=\lambda_{2k-1}(R)=\lambda_{2k}(R)$. For the third inequality to be saturated we have to choose $\lambda_{d}(R)=0$ in odd dimensions.
But this implies $S_{\text{R}}\ge 1/\sqrt{d-1}$ so that the bound cannot be tight for $S_{\text{R}}<1/\sqrt{d-1}$.
Now choose the smallest eigenvalue of $R$ to be $\lambda_{d}(R)=t\le1$, it immediately follows that $\Tr R^2\ge d(d-2t+t^2)/(d-1)$ where the minimum is attained for $\lambda_{k}(R)=(d-t)/(d-1)$.
For a given $\Tr R^2$ it must therefore hold that 
\begin{align*}
    t\ &\ge\ 1-\frac{\sqrt{(d-1)\Tr R^2 -d^2+d}}{\sqrt{d}}\\
    &=\ 1-\sqrt{d-1}S_{\text{R}}\ .
\end{align*}
It now holds in odd dimensions for $S_{\text{R}}\le1/\sqrt{d-1}$,
\begin{align*}
    dS_{\text{I}}^2&=\Tr I^2=\sum\lambda_k(I)^2\le\sum\limits_{k=1}^{d-1}\lambda_{k}(R)^2=\Tr R^2-t^2\\
    &\le d+dS_{\text{R}}^2-(1-\sqrt{d-1}S_{\text{R}})^2=(\sqrt{d-1}+S_{\text{R}})^2,
\end{align*}
which proves the result.

\end{proof}

The inequality~\eqref{eq:linear_bound} for odd dimensions and $S_{\text{R}}\le 1/\sqrt{d-1}$ is saturated by states of the form
\begin{align}\label{states_odd_2}
    \rho\ =
    \scalebox{0.75}{$
    \begin{pmatrix} 
    \alpha & -i\alpha &&&&&&\\
    i\alpha & \alpha &&&&&&\\
    && \alpha & -i\alpha  &&&&\\
    && i\alpha  & \alpha &&&&\\
    &&&& \ddots &&&\\
    &&&&& \alpha & -i\alpha  &\\
    &&&&& i\alpha  & \alpha &\\
    &&&&&&& 1-(d-1)\alpha
    \end{pmatrix}
    $}
\end{align}
with $1/(d-1)\ge\alpha\ge1/d$. To see this it suffices to note that  $\Tr R^2=d^2(d(d-1)\alpha^2-2(d-1)\alpha+1)$
and \mbox{$\Tr I^2=d^2(d-1)\alpha^2$}.
Figures~\ref{fig:areas}--\ref{fig:asymptotic} show this bound in green. 

For even dimension $d>2$, Eqs.~\eqref{eq:purity_bound} and~\eqref{eq:quadratic_bound} describe the boundary, they intersect at $S_{\text{R}}=\sqrt{(d-2)/2}$. For odd dimension the boundary 
(for $S_{\text{R}}\le 1/\sqrt{d-1}$)  is described by Eq.~\eqref{eq:linear_bound}, it tangentially touches the bound in Eq.~\eqref{eq:quadratic_bound} at $S_{\text{R}}= 1/\sqrt{d-1}$.

Going back to our original coordinates $S_{\text{D}}$, $S_{\text{X}}$ and $S_{\text{I}}$, we immediately find the following result.

\begin{coro}\label{coro:full_bounds}
    The three quantities $S_{\mathrm{D}}$, $S_{\mathrm{X}}$ and $S_{\mathrm{I}}$ satisfy the bounds 
    \begin{subequations}
    \begin{align}
        S_{\mathrm{I}}^2&\ \le\ 1 +S_{\mathrm{D}}^2+S_{\mathrm{X}}^2\\
       \sqrt{d} S_{\mathrm{I}}&\ \le\ \sqrt{d-1} +\sqrt{S_{\mathrm{D}}^2+S_{\mathrm{X}}^2}\ \ ,
    \end{align}
    \end{subequations}
    where the last inequality holds only for $S_{\mathrm{D}}^2+S_{\mathrm{X}}^2\le1/(d-1)$.
\end{coro}

\begin{proof}
    The bounds follow trivially from Eq.~\eqref{eq:quadratic_bound} and~\eqref{eq:linear_bound} by replacing $S_{\text{R}}^2=S_{\text{D}}^2+S_{\text{X}}^2$.
\end{proof}

\begin{prop}
    The bounds in Corollary~\ref{coro:full_bounds} define a tight upper bound. Together with Eq.~\eqref{eq:purity_bound} they are optimal, in the sense that there exist no further constraints on the quantities $S_{\mathrm{D}}$, $S_{\mathrm{X}}$ and $S_{\mathrm{I}}$ that do not trivially follow from them.
\end{prop}

\begin{proof}
The states in Eq.~\eqref{states_even}, \eqref{states_odd_1} and \eqref{states_odd_2} have no real off-diagonal contribution and thus saturate the bounds in Corollary~\ref{coro:full_bounds} for $S_{\text{X}}=0$, together with the pure states $\beta\ket{0}+i\sqrt{1-\beta^2}\ket{1}$, where $0\leq\beta\leq 1$.

We will now define a continuous orthogonal transformation that shifts weight from the diagonal component to the real off-diagonal component as long as $S_{\text{D}}>0$. Note that any  orthogonal (i.e., real) transformation leaves the coordinates $S_{\text{R}}$ and  $S_{\text{I}}$ unchanged. 
Consider an orthogonal transformation $O_{kl}(\theta)= \cos{\theta}(\ketbra{k}{k}+\ketbra{l}{l})+\sin{\theta}(\ketbra{k}{l}-\ketbra{l}{k})$ acting on the subspace spanned by the states $\{\ket{k},\ket{l}\}$. 
Start acting with this transformation on the first and the last entry until either one of the new diagonal entries becomes $1/d$.
If it is the last entry we perform a second orthogonal transformation from $\ket{0}$ to $\ket{d-2}$ until one of them reaches $1/d$. If it is the first entry we apply the transformation from $\ket{1}$ to $\ket{d-1}$. We repeat this process until all diagonal entries are equal to $1/d$.
This procedure continuously shifts weight from the diagonal coordinate $S_{\text{D}}$ onto the real off-diagonal coordinate $S_{\text{X}}$ until the diagonal coordinate vanishes. The orbits of the states in Eqs.~\eqref{states_even}, \eqref{states_odd_1} and \eqref{states_odd_2} form the boundary.
This concludes the proof.

\end{proof}

All inequalities on the quantum state space follow from the positivity (and normalization) of quantum states. However, to check for positivity becomes increasingly hard with increasing dimension. Moreover, those conditions offer little insight into the structure of the quantum state space. The inequalities in Corollary~\ref{coro:full_bounds} of the newly defined coordinates give us necessary (but not sufficient) conditions for the positivity of a quantum state. They introduce a novel bound on the magnitude of the imaginary part of a quantum state that is easy to verify and at the same time provides new insight on the structure of the quantum state space.

\begin{figure}
    \includegraphics[width=0.9\columnwidth]{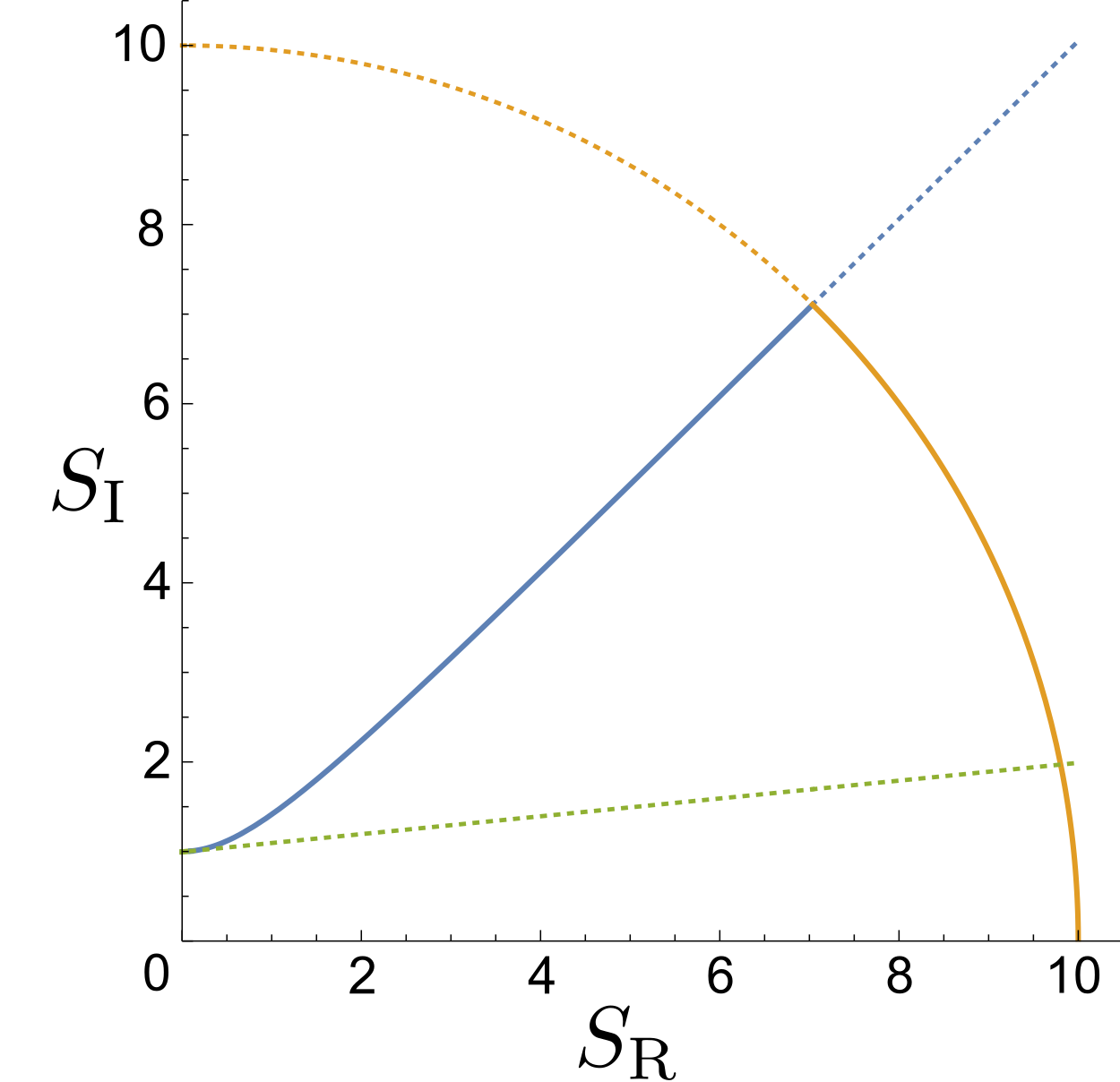}
    \caption{The figure shows the bound of Eq.~\eqref{eq:purity_bound} in orange, the bound of Eq.~\eqref{eq:quadratic_bound} in blue and the bound of Eq.~\eqref{eq:linear_bound} in green for dimension $d=101$.
    }
    \label{fig:asymptotic}
\end{figure}

\section{Bloch ball for a qudit}

We now discuss  the geometric implications
of our investigation. In recent years considerable effort was devoted to the question of describing the quantum state space of dimension $d$~\cite{Kimura_2003,Bengtsson_2012} and constructing models for the quantum state space , e.g., Refs.~\cite{Eltschka_2021,DCM29429,Levitt_2021,Rau_2021,jarov2023,Sharma_2024,Morelli_2024,shravan2023geometry}. A model of the quantum state space represents the set of all quantum states and is therefore distinct from visualizations of individual states, like for example presented in Ref.~\cite{Kurzynski_2016}. That is, one asks for the shape of this $(d^2-1)$-dimensional object if it is represented in a lower dimension, thereby maintaining as many geometric properties of the original state space as possible. Because of the reduction in the number of parameters for dimensions $d\geq 3$ it is clear that any model, while correctly displaying some of the properties of the state space, will suffer from a loss of information and inadequately represent other features. Consequently, any model is the result of a judicious choice of the coordinates, depending on the set of properties that are to be maintained. The corresponding objects may have rather different shapes.

Our work provides a striking illustration for these considerations. As is evident from Fig.~\ref{fig:areas3d}, our parametrization of the qudit state Eqs.~\eqref{eq:DXImat}, \eqref{eq:purity} in imaginary, real diagonal and off-diagonal coordinates leads to a three-dimensional model for the qudit state space for all $d\geq 2$.
By construction, this model preserves the norm of the Bloch vector.
Other models that fall into this category have been introduced in Refs.~\cite{WyderkaGuhne20,Eltschka_2021,Sharma_2024,Morelli_2024}.
In all these models the vector space encompassing the state space is divided into three orthogonal subspaces. The norms of the Bloch vector in each subspace are then used as coordinates.
The difference between the models is determined by the choice of orthogonal subspaces. If we choose a matrix basis that splits along the chosen subspaces, e.g., the Gell-Mann matrices, then we can regard this as splitting the Bloch vector into three parts with the corresponding norms of each part as coordinates.
The model of a qutrit in Ref.~\cite{Eltschka_2021} emphasizes the simplex of diagonal states by keeping a faithful representation of this two-dimensional subspace.
The two-qubit model in Ref.~\cite{Morelli_2024} and the three-qubit model in Ref.~\cite{WyderkaGuhne20} focus on distinguishing local and global properties, the coordinates capture one-body marginals and two- or three-body correlations.
The present model focuses on the bounds for the imaginary coordinate with respect to the real (diagonal and off-diagonal) counterparts.

In contrast to the two models in Refs.~\cite{Eltschka_2021,Morelli_2024}, this model is not convex for $d\geq 3$. This is not surprising, the norm is a nonlinear function and there is generally no reason to expect that the convex state space is mapped to a convex region. A similar behavior is observed in  Ref.~\cite{WyderkaGuhne20} (for the non-squared coordinates).

Another salient feature of the present model is its rotation invariance about the imaginary axis, cf.~Fig.~\ref{fig:areas3d}. It is a direct consequence of the unitary invariance of the purity, Eq.~\eqref{eq:purity}: Any transformation that leaves $S_{\text{I}}$ unchanged, necessarily preserves also the radius $S_{\text{R}}$ in the real plane.

Intriguingly, the bounds in Eqs.~\eqref{eq:quadratic_bound}, \eqref{eq:linear_bound} reveal that the geometry of the state space boundary is different in even vs.\ odd dimensions: While for odd dimensions there is always a flat surface part (of mixed states) close to the most distant point on the imaginary axis, the surface is curved in the corresponding region for even dimensions, cf.~Fig.~\ref{fig:areas3d} (b).
We note that the concave part of the model in $d=3$ is completely described by the linear bound Eq.~\eqref{eq:linear_bound}. 
For this special case convexity is restored if just one real coordinate $S_{\mathrm{R}}$ is used (instead of $S_{\mathrm{X}}$ and $S_{\mathrm{D}}$), as can be observed in Fig.~\ref{fig:areas} (a).
That $d=3$ plays a special role has no deep reason, but is a simple consequence of the low dimension. There is only one unique state of the form Eq.~\eqref{states_odd_1}, which is pure, and there is no quadratic boundary. This is in analogy with the fact that for $d=2$ the boundary is completely described by the purity constraint Eq.~\eqref{eq:purity_bound}, as there is only one unique (pure) state of the form Eq.~\eqref{states_even}.
Hence the first dimension for which the generic odd-$d$ behavior is observed is $d=5$: In the vicinity of $S_{\text{I}}=\sqrt{(d-1)/d}$ the linear bound describes the surface before the quadratic bound, Eq.~\eqref{eq:quadratic_bound}, takes effect for $S_{\text{R}}=1/\sqrt{d-1}$.  The linear surface part remains present for all odd dimensions but becomes negligibly small for large $d$, see Fig.~\ref{fig:asymptotic}.
At the same time the restriction imposed by Eq.~\eqref{eq:quadratic_bound} becomes more stringent. The respective part in $\left(S_{\text{R}},S_{\text{I}}\right)$ coordinates that remains compatible with a quantum state decreases, until asymptotically only half of the original circular area defined by the purity bound in Eq.~\eqref{eq:purity} remains, see Fig.~\ref{fig:asymptotic}.

\section{Relevant Connections}

Recently, the resource theory of imaginarity has received considerable attention~\cite{Hickey_2018, Wu_2021,Wu_2021_2} 
and one might ask how our results relate to it.
A measure of imaginarity is a quantity that is non-increasing under real quantum channels, i.e., completely positive and trace-preserving (CPTP) maps that allow a description with purely real Kraus operators.
The robustness of imaginarity is defined as $\mathcal{I}(\rho)=\|\rho-\rho^T\|_1/2=\|I/d\|_1$ and the geometric measure of imaginarity is defined as $\mathcal{I}_G(\psi)=1-\sup_{\phi\in\mathcal{R}}|\langle\phi|\psi\rangle|^2$, with the convex roof extension for mixed states~\cite{Wu_2021_2}.
While these quantities are imaginarity monotones, our imaginary coordinate $S_{\text{I}}$ is not an imaginarity measure~\cite{Chen_2023}. This stems from the fact that the Hilbert-Schmidt norm is generally not contractive under CPTP maps for dimension larger than two~\cite{Ozawa_2000}, in contrast to the trace norm~\cite{Ruskai1994BEYONDSS}.
Our goal is not the study of imaginarity as a resource, rather we characterize the constraints on coordinates describing the quantum state space.
Nonetheless, our analysis establishes a relation between the real weight and the robustness of imaginarity $\mathcal{I}$ for odd dimensions. For $S_{\text{R}}\geq1/\sqrt{d-1}$ there exist states that achieve $\mathcal{I}=1$ (cf.~the states saturating the bound in Proposition~\ref{prop:i<1+r}), while for
$0\le S_R<1/\sqrt{d-1}$, it holds that $\mathcal{I}\le1-1/d+\sqrt{d-1}/d\ S_{\text{R}}<1$, saturated by the states in Proposition~\ref{prop:linear_bound}.
On the other hand, for even dimensions, no restriction exists concerning real weight and imaginarity, showing again remarkable differences between even and odd dimensions.
As a consequence, in even dimension the purity of states with maximal imaginarity $\mathcal{I}=1$ is bounded by $\Tr{\rho^2}\ge2/d$, whereas in odd dimensions by $\Tr{\rho^2}\ge 2/(d-1)$.

Apart from the general objective of gaining deeper insight into the geometric structure of the quantum-mechanical state space, there are also clear physics-driven motivations for understanding this structure. A natural application of such models is the visualization of the time evolution of quantum states, in particular for dissipative systems -- or, in other words, the visualization of the action of quantum channels on higher-dimensional systems, in analogy with qubit dynamics in the standard Bloch ball (see, for example, Refs.~\cite{Feynman1957,NielsenChuang10,Bengtsson_2006, Levitt_2021,Eltschka_2021,Cafaro_2025}). Another interesting application is discussed in Ref.~\cite{Levitt_2021}, where the authors investigate regions of the state space of nuclear spins in which hyperpolarization becomes possible in nuclear magnetic resonance experiments at a given temperature, based on an entropic criterion.
These examples illustrate that there are strong and diverse motivations for studying such geometric structures.

\section{Conclusion}

We generalize the Bloch coordinates of a qubit describing the diagonal, real off-diagonal and imaginary part of a quantum state to arbitrary finite dimensions. We completely characterize the set of coordinates compatible with a quantum state. While for the simplest quantum system of a qubit these coordinates are interchangeable, this property is lost in higher dimensions. The imaginary coordinate plays a distinctive role and is bounded with respect to the other coordinates.
With these findings we obtain a three-dimensional model of the generally higher-dimensional quantum state space that complements the already established models.
Surprisingly, our bound exhibits a qualitative difference between even and odd dimensions, where the general behavior for odd dimensions is first observed in dimension 5. 
We would like to mention that the highly topical open questions about the existence of mutually unbiased bases~\cite{DURT_2010} and symmetric informationally complete measurements~\cite{ScottGrassl_2010} are in essence the question whether the shape of the quantum state space (in the realm of pure states) qualitatively changes with the dimension.
Our results provide an indication that for mixed states such differences indeed do exist.

\section{Acknowledgements}

The authors thank A. Izquierdo for discussions.
S.M. is a recipient of an APART-MINT Fellowship of the Austrian Academy of Sciences at the Atominstitut of the Technische Universit\"at Wien and
was supported by
BCAM-IKUR, funded by the Basque Government by the IKUR Strategy and by the European Union NextGenerationEU/PRTR,
the Basque Government through the BERC 2022-2025 program
and by the Ministry of Science and Innovation: BCAM Severo Ochoa accreditation CEX2021-001142-S / MICIN / AEI / 10.13039/501100011033.
S.L. was funded in whole or in part by the National Science Centre, Poland 2024/54/E/ST2/00451, by the Polish National Agency for Academic Exchange under the Strategic Partnership Programme grant BNI/PST/2023/1/00013/U/00001, from the Spanish MICINN with funding from European Union NextGenerationEU (PRTRC17.I1) and the Generalitat de Catalunya, and from the Spanish Agencia Estatal de Investigaci\'on and the MICINN (grant PID2022-141283NB-I00).
J.S. was supported by
Grant PID2021-126273NB-I00 funded by MCIN/AEI/10.13039/501100011033 and by "ERDF A way of making Europe" as well as by
the Basque Government through Grant No.\ IT1470-22.
For the purpose of Open Access, the authors have applied a CC-BY public copyright licence to any Author Accepted Manuscript (AAM) version arising from this submission.
The authors acknowledge TU Wien Bibliothek for financial support through its Open Access Funding Programme.

\bibliographystyle{bibstyle2}
\bibliography{bibfile}

\end{document}